\documentclass[a4paper,UKenglish,cleveref, autoref, thm-restate]{lipics-v2021}
\usepackage{booktabs}

\usepackage[ruled,linesnumbered,noend]{algorithm2e}
\usepackage{lipsum,graphicx,subcaption}
\usepackage[dvipsnames,table,xcdraw]{xcolor}
\usepackage{mathtools}

\DeclarePairedDelimiter{\ceil}{\lceil}{\rceil}

\title{Thinking Out of the Box: Hybrid SAT Solving by Unconstrained Continuous Optimization} %TODO Please add

\author{Zhiwei Zhang}{Department of Computer Science, Rice University}{zhiwei@rice.edu}{}{}%TODO mandatory, please use full name; only 1 author per \author macro; first two parameters are mandatory, other parameters can be empty. Please provide at least the name of the affiliation and the country. The full address is optional. Use additional curly braces to indicate the correct name splitting when the last name consists of multiple name parts.

\author{Samy Wu Fung}{Department of Applied Mathematics and Statistics, Colorado School of Mines}{swufung@mines.edu}{}{}

\author{Anastasios Kyrillidis}{Department of Computer Science, Rice University}{anastasios@rice.edu}{}{}

\author{Stanley Osher}{Department of Mathematics, University of California, Los Angeles}{sjo@math.ucla.edu}{}{}

\author{Moshe Y. Vardi}{Department of Computer Science, Rice University}{vardi@rice.edu}{}{}

\authorrunning{Z. Zhang, S. Wu Fung, A. Kyrillidis, S. Osher, M. Y. Vardi}
\Copyright{Jane Open Access and Joan R. Public} %TODO mandatory, please use full first names. LIPIcs license is "CC-BY";  http://creativecommons.org/licenses/by/3.0/


\category{} %optional, e.g. invited paper

\relatedversion{} %optional, e.g. full version hosted on arXiv, HAL, or other respository/website
\nolinenumbers %uncomment to disable line numbering

\EventEditors{John Q. Open and Joan R. Access}
\EventNoEds{2}
\EventLongTitle{42nd Conference on Very Important Topics (CVIT 2016)}
\EventShortTitle{CVIT 2016}
\EventAcronym{CVIT}
\EventYear{2016}
\EventDate{December 24--27, 2016}
\EventLocation{Little Whinging, United Kingdom}
\EventLogo{}
\SeriesVolume{42}
\ArticleNo{23}
\begin{document}

\maketitle

\begin{abstract}
%Boolean satisfiability (SAT) problems are an important class of computational problems that appear in many domains such as including hardware verification, planning, and automated reasoning. Traditional methods use discrete solvers, but recent research formulates the SAT problem as a constrained continuous optimization problem. This paper furthers this line of research and proposes an unconstrained optimization formulation of the SAT problem. We discuss the trade-offs of theoretical properties compared with the constrained version and give a theoretical description of the types of constraints that are amenable to unconstrained settings. Experiments show that the unconstrained formulation is a promising approach.

The Boolean satisfiability (SAT) problem lies at the core of many applications in combinatorial optimization, software verification, cryptography, and machine learning. While state-of-the-art solvers have demonstrated high efficiency in handling conjunctive normal form (CNF) formulas, numerous applications require non-CNF (hybrid) constraints, such as XOR, cardinality, and Not-All-Equal constraints. Recent work leverages polynomial representations to represent such hybrid constraints, but it relies on box constraints that can limit the use of powerful unconstrained optimizers. In this paper, we propose unconstrained continuous optimization formulations for hybrid SAT solving by penalty terms. We provide theoretical insights into when these penalty terms are necessary and demonstrate empirically that unconstrained optimizers (e.g., Adam) can enhance SAT solving on hybrid benchmarks. Our results highlight the potential of combining continuous optimization and machine-learning-based methods for effective hybrid SAT solving.
\end{abstract}

\section{Introduction}
The Boolean satisfiability (SAT) problem \cite{Vardi14a} asks whether there exists a solution that satisfies all constraints in a given set of Boolean constraints. This fundamental problem holds immense significance in computer science with applications spanning combinatorial optimization \cite{horbach2010boolean}, software verification \cite{SATSolvingVerification}, probabilistic inference \cite{SATprobabilisticInference}, mathematical conjecture proving \cite{heule2016solving}, machine learning~\cite{katz2017reluplex}, and quantum computing~\cite{soeken2020boolean,vardi2023solving}. While SAT is known to be NP-complete, recent decades have witnessed remarkable advances in SAT solver technology \cite{Vardi14a} for both CDCL-based complete solvers ~\cite{silva1996grasp, bayardo1997using} and heuristic-search incomplete solvers \cite{GSAT}. 

The landscape of complete SAT solvers is dominated by Conflict-Driven Clause Learning (CDCL)~\cite{silva1996grasp, bayardo1997using} approaches, which evolved from the seminal GRASP algorithm \cite{marques1999grasp}. CDCL represents a significant advancement over the earlier backtracking Davis-Putnam-Logemann-Loveland (DPLL) algorithm \cite{davis1962machine}. This paradigm has spawned numerous high-performance implementations, including groundbreaking tools like Chaff \cite{chaff_paper}, MiniSat \cite{minisat}, and Glucose \cite{glucose}, alongside more contemporary developments such as MapleSAT \cite{maplesat} and Kissat \cite{fleury2020cadical}. The success of CDCL-based methods has established them as the predominant approach in SAT-solving research. Incomplete SAT solvers primarily employ discrete local search (DLS) \cite{hoos2000local} and message passing (MP) \cite{braunstein2005survey}, both optimizing to minimize unsatisfied constraints. Important incomplete solvers include GSAT \cite{GSAT}, WSAT \cite{walksat}, and newer innovations like TaSSAT~\cite{chowdhury2024tassat}, WalkSATlm \cite{Walksat-implementation}, probSAT \cite{probSAT}, Dimetheus \cite{kader2017novel}, and Sparrow \cite{sparrow}. Despite lacking comprehensive guarantees, incomplete solvers are highly efficient for random and crafted problem classes \cite{GSAT}.

While practical SAT solvers have achieved remarkable success, most of them are limited to solve a specific problem format named conjunctive normal form (CNF)\cite{schuler2005algorithm}, yet many applications require more complex logic. Neural network verification needs pseudo-Boolean constraints \cite{paredes2019principled}, while cryptography relies on XOR constraints \cite{Biclique-Cryptanalysis-of-the-Full-AES}  as well as cardinality constraints (\texttt{CARD}) and Not-all-equal (\texttt{NAE}) constraints in discrete optimization \cite{Graph-coloring-with-cardinality-constraints,nae-coloring}. Approaches for handling non-CNF constraints include \texttt{CNF} encodings \cite{cnfencodingofXOR,pblib} (which vary in effectiveness based on characteristics \cite{encoding-handbook-of-satisfiability}) and direct solver extensions like CryptoMiniSAT \cite{cmspaper} for \texttt{XOR} constraints. These constraint-specific solutions highlight the need for unified approaches to handling diverse hybrid constraints. Using hybrid constraints makes the SAT formula significantly more succinct and readable in those applications during the encoding phase. As for solving the hybrid formula, while CNF encoding \cite{prestwich2009cnf} and extensions \cite{yang2021engineering} to CNF SAT solvers can accept hybrid constraints, those methods do not natively handle hybrid constraints and can often be inefficient. 

 A promising research direction has emerged that approaches SAT solving with hybrid constraints by continuous optimization.
 %\myv{You need to define and motivate hybrid constraints through {continuous optimization}. 
 This idea was pioneered in the 1990s~\cite{gu1999optimizing} using polynomial representation for CNF formulas, though this early work was limited to CNF format exclusively. Recently, a general approach named FourierSAT~\cite{fouriersat,kyrillidis2021solving,msthesis} transforms multiple types of hybrid constraints into polynomials via the Walsh-Hadamard-Fourier expansion. This formulation recasts the SAT problem as a constrained optimization problem with a polynomial objective and box constraints in $[-1,1]^n$. Notably, solutions to the original SAT problem correspond to global optima in this optimization formulation.

FourierSAT's core advantage lies in its ability to handle non-CNF formulas effectively. Additionally, it offers interesting theoretical properties through analysis of polynomials and continuous optimizers \cite{nocedal2006numerical}. Subsequent research has enhanced gradient computation efficiency through specialized data structures \cite{gradsatAAAI} and GPU acceleration \cite{cen2023massively}. Continuous-optimization-based SAT solvers are particularly promising due to their ability to leverage advances from machine learning optimization techniques, which can be further accelerated using specialized hardware like GPUs and TPUs \cite{nikolic2022survey}.

% machine learning's neural-network-training techniques. The ML community has developed numerous successful algorithms for non-convex objective functions, including stochastic gradient descent (SGD)~\cite{amari1993backpropagation}, momentum methods~\cite{polyak1964some}, and adaptive-step-size approaches like Adam~\cite{kingma2014adam}. These methods can be further accelerated using specialized hardware such as GPUs and TPUs~\cite{nikolic2022survey}.

However, a potential limitation of FourierSAT stems from its box constraints, particularly when rapid solutions are required. These constraints restrict the application of powerful optimization algorithms such as second-order methods~\cite{nocedal2006numerical} and unconstrained global optimization methods~\cite{heaton2024global}. 
For instance, projected Newton methods using Euclidean projection encounter convergence challenges~\cite{gafni1984two, bertsekas1982projected, kan2021pnkh}. 

\noindent \textbf{Contribution}. We propose an unconstrained formulation of FourierSAT and explore its potential for hybrid SAT solving. Our approach transforms the constrained optimization problem by applying a non-negative operator to the polynomial representation of each constraint. These transformed constraints become additional penalty terms in our formulation and ensure the existence of a global minimum in $\mathbb{R}^n$ and a direct correspondence between solutions of the unconstrained problem and the original constrained problem. Our theoretical analysis establishes a necessary condition for omitting the penalty terms without sacrificing the soundness of the reduction. 
%In particular, our experiments show that 
% We show that while CNF and XOR constraints can be effectively handled without such penalty terms, 
%cardinality constraints fundamentally require them to maintain algorithmic validity.

By implementing our methods within the FourierSAT codebase and evaluating them on an extensive benchmark set, we uncover several important insights: (1) the added penalty terms effectively help satisfy cardinality constraints while hindering the satisfaction of CNF and XOR constraints; (2) different formulations produce varying effects across benchmarks and solver configurations; and (3) adopting unconstrained optimizers such as Adam~\cite{kingma2014adam} significantly improves the performance of the virtual best solver.

This work enriches both the theoretical understanding and practical application of continuous-optimization-based SAT solving techniques. It opens up promising avenues for applying general machine-learning-inspired unconstrained optimizers to hybrid SAT solving.

\section{Notation and Preliminaries}
\label{section:pre}
\subsection{Boolean Formulas, Constraints}

\begin{definition}[Boolean constraints]
    Let $X=(x_1,...,x_n)$ be a sequence of $n$ Boolean variables.  An assignment is a vector $b\in\{\texttt{True},\texttt{False}\}^n$. A Boolean constraint  $c$ is a mapping from Boolean assignments $\{\texttt{True},\texttt{False}\}^n$ to $\{\texttt{True},\texttt{False}\}$. 
\end{definition}

\begin{definition}[Boolean formulas, SAT and MaxSAT]  A formula $f=c_1\wedge c_2\wedge \dots \wedge c_m$ is the conjunction of $m$ Boolean constraints. A solution to $f$ is an assignment that satisfies all its constraints. The Boolean Satisfiability problem (SAT) is to decide whether a formula is satisfiable. When the formula is unsatisfiable, the MaxSAT problem asks for the maximum number of constraints satisfied by an assignment.
\end{definition}

When all constraints in a formula are \texttt{OR}s, i.e., disjunctions of a set of variables or its negations, the formula is said to be in conjunctive normal form (CNF). While CNF is the most prevalent format in SAT solving, this work focuses on solving \emph{hybrid} constraints, i.e., constraints can be of multiple types such as \texttt{OR}, XOR, not-all-equal (NAE), cardinality, etc.
	
	\subsection{Walsh-Fourier Expansion of a Boolean Function}
	%We define a Boolean function by $f: \{\pm 1\}^n \to \{0,1\}$, where for the variable value, $-1$ is used to stand for \texttt{True} and $+1$ for \texttt{False}. An assignment now is a vector $a\in\{ \pm 1\}^n$.
	
    The Walsh-Fourier Transform is a method for transforming a Boolean function into a multilinear polynomial. The following theorem states that every function defined on a Boolean hyper-cube has an equivalent polynomial representation.  

\begin{definition} A multilinear polynomial is a multivariate polynomial that is linear in each of its variables, i.e., each monomial is a constant times a product of distinct variables.
\end{definition}
For example, $f(x,y,z)=3xy+4y-5z$ is multilinear, but $f(x,y,z)=x^2+3y$ is not.

	\begin{theorem}[Walsh-Fourier expansion] \rm{\cite{O'Donnell:2014:ABF:2683783}} \label{FourierTransformation} Given a function $f: \{\pm 1\}^n \to \mathbb{R}$, there is a unique way of expressing $f$ as a multilinear polynomial with at most $2^n$ terms, where each term corresponds to one subset of $[n]$, according to: $	f(x) = \sum_{S\subseteq [n]} \left( \widehat{f}(S) \cdot \prod_{i\in S}x_i \right),$
		where $\widehat{f}(S)\in \mathbb{R}$ is called a Fourier coefficient on $S$, and computed as:
		       $ \widehat{f}(S) = \underset{x\sim \{\pm 1\}^n}{\mathbb{E}} \left[f(x) \cdot \prod_{i\in S}x_i \right] 
		        = \frac{1}{2^n} \!\!\! \sum_{x\in \{\pm 1\}^n} \left(f(x) \cdot \prod_{i\in S}x_i\right)$,
		where $x\sim \{\pm 1\}^n$ indicates $x$ is uniformly sampled from $\{\pm 1\}^n$.
		The polynomial is called the \textbf{Walsh-Fourier expansion} of $f$. 
	\end{theorem}
	
\begin{table}
    \centering
    \begin{small}
    \begin{tabular}{c c c}
        \toprule 
        $c$ & & $\texttt{FE}_c$ \\
        \cmidrule{1-1} \cmidrule{3-3}
        $x_1 \vee x_2$ & & $\displaystyle 1/4\Bigl(1 + x_1 + x_2 + x_1x_2\Bigr)$  \\ 
        $x_1 \oplus x_2 \oplus x_3$ & & $\displaystyle 1/2\Bigl(1 + x_1x_2x_3\Bigr)$  \\ 
        $\text{\texttt{NAE}}(x_1,x_2,x_3)$ & & $\displaystyle 1/4\Bigl(1 + x_1x_2 + x_2x_3 + x_1x_3\Bigr)$  \\ 
        \bottomrule
    \end{tabular} 
    \caption{Examples of Fourier expansions of \texttt{OR}, \texttt{XOR} and \texttt{Not-All-Equal} constraints, where $x_i\in\{+1,-1\}$, with $+1=$ false and $-1=$ true. \texttt{FE}$_c$ yields 0 when the constraint is satisfied and 1 otherwise.}
    \label{ex_FE}
    \end{small}
\end{table}
	
	For a Boolean constraint $c$, $\texttt{FE}_c:\{-1,1\}^n\to\{0,1\}$ denotes its Walsh-Fourier expansion. Table \ref{ex_FE} lists examples of Walsh-Fourier expansions for different types of constraints. Note that for the variable of variables, $-1$ stands for \texttt{True} and $1$ for False. For the value of the constraint, $0$ represents \texttt{True}, and $1$ represents \texttt{False} \footnote{Note that the numerization of the output in this paper ($0$ for \texttt{True} and $1$ for \texttt{False}) is different from the original FourierSAT paper \cite{fouriersat}. The numerization used in this paper can be understood as the cost indicator of a Boolean constraint. We believe this numerization is best for simplifying formulations and unifying constrained and unconstrained version of FourierSAT. Meanwhile, different numerization do not essentially affect the theoretical properties.}.

\section{Formulations of SAT by Continuous Optimization }
In FourierSAT \cite{fouriersat}, hybrid Boolean SAT/MaxSAT problems are converted to constrained continuous optimizations. In this section, we first recap the constrained formulation in FourierSAT with box constraint $x\in [-1,1]^n$ and then propose unconstrained formulations.

%\textcolor{blue}{Zhiwei: I am thinking about adding a running example for all the formulations on the same formula with 2 variables. Then we can write down the formulation as well as plot its landscape.}
\subsection{Constrained formulation}
\begin{definition}[Constrained formulation]
Given a formula $f$ with constraint set $C$, the constrained (linear) formulation is defined as $$F^{lin}_C(x)= \sum_{c\in C} \texttt{FE}_c(x).$$
\end{definition}

The following proposition gives the correctness of FourierSAT. We define a sign function on the domain of real vectors, $\texttt{sgn}:\mathbb{R}^n\to \{-1,1\}^n$ by
$
\texttt{sgn}(l)_i =
\begin{cases}
-1 & \text{if } l_i < 0, \\
\phantom{-}1 & \text{if } l_i \ge 0,
\quad \text{for } i = 1, \dots, n.
\end{cases}
$

\begin{proposition} \rm{\cite{fouriersat}}
A hybrid formula $f$ is satisfiable iff  $$\min_{x\in[-1,1]^n}F_{C}^{lin}(x) = 0.$$ Moreover if $l\in [-1,1]^n$ and $F_{C}^{lin}(l)=0$, then \texttt{sgn}$(l)$ satisfies $f$.
\label{prop:fouriersat}
\end{proposition}

Proposition \ref{prop:fouriersat} is a special case of the following Proposition, which states that the solutions to a MaxSAT problem, i.e., discrete assignments with the minimum number of satisfied constraints, are also preserved by $F_C^{lin}$.

\begin{proposition} \rm{\cite{fouriersat} }If $k$ is the minimum number of violated constraints in $C$ by all Boolean assignments, then
$
\mathop{min}\limits_{x\in[-1,1]^n}F^{lin}_C(x) = k.
$  Moreover if $l\in [-1,1]^n$ and $F_{C}^{lin}(l)=k$, then \texttt{sgn}$(l)$ is the solution to the MaxSAT problem.
\label{prop:maxsat}
\end{proposition}

\subsection{Unconstrained Formulations}
In this work, we extend the search space of the continuous optimization in FourierSAT from $[-1,1]^n$ to $\mathbb{R}^n$ with two modifications to the constrained formulations. First, non-negative operators are applied to $\texttt{FE}_c$ so that the value contributed by each constraint is lower-bounded. Second, penalty terms for the box constraints are used to force the optimizer converging to Boolean points.  For brevity, we leave all the proofs in the appendix.

\begin{definition}[Square formulation]
Given a formula $f$ with constraint set $C$, the square formulation w.r.t. $\alpha$ ($\alpha\ge 0$) is defined as $$F^{sq}_{C,\alpha}(x) = \sum_{c\in C} \texttt{FE}_c^2(x) + \alpha\cdot \sum_{i\in [n]}(x^2_i-1)^2.$$ 
\end{definition}

% \begin{proposition}
% A hybrid formula $f$ is satisfiable iff $\mathop{min}\limits_{x\in \mathbb{R}^n}F^{sq}_{C,\alpha}(x)=0$ for $\alpha>0$.
% \label{prop:sq}
% \end{proposition}

% \begin{proof}
%     $\Rightarrow$: Suppose $f$ is satisfiable, then there exists an assignment $b\in\{-1,1\}^n$ such that  $FE_c(b)=0$ for all $c\in C$. Therefore $FE_{sq,\alpha}(b)=0$. Since $FE_{sq,\alpha}FE_c\ge 0$ on all real vectors, we have $\mathop{min}\limits_{x\in \mathbb{R}^n}F_{sq,\alpha}=0$.

%     $\Leftarrow$: Suppose $\mathop{min}\limits_{x\in \mathbb{R}^n}F_{sq,\alpha}=0$. Then there exists a real vector $l\in \mathbb{R}^n$ such that $F_{sq,\alpha}(l)=0$. Since $F_{sq,\alpha}$ is in the form of sum of the square, then we have $FE_c(l)=0$ for all $c\in C$. We also have $(l_i^2-1)=0$ for all $i\in [N]$. Therefore $l$ is a Boolean assignment and by proposition \ref{}, $l$ is a solution to the formula $f$ and $f$ is hence satisfiable.
% \end{proof}
The square formulation provides a sound reduction from SAT to continuous optimization.
\begin{proposition}[Soundness]
A hybrid formula $f$ with constraint set $C$ is satisfiable iff $$\mathop{min}\limits_{x\in \mathbb{R}^n}F^{sq}_{C,\alpha>0}(x)=0.$$ Moreover if $l\in \mathbb{R}^n$ and $F_{C,\alpha}^{sq}(l)=0$, then \texttt{sgn}$(l)$ satisfies $f$.
\label{prop:sq}
\end{proposition}
% \noindent The detailed proof of Proposition~\ref{prop:sq} is given in Appendix~\ref{app:proof_sq}.
Alternatively, one can also use \texttt{abs} as non-negative operator.

 \begin{definition}[Absolute-value formulation]
   Given a formula $f$ with constraint set $C$, define the absolute-value formulation as $$F^{abs}_{C,\alpha}(x) = \sum_{c\in C} |\texttt{FE}_c(x)| + \alpha\cdot \sum_{i\in [n]}(x_i^2-1)^2.$$
 \end{definition}

Proposition \ref{prop:abs} can be proved similarly with Proposition \ref{prop:sq}.
\begin{proposition}[Soundness]
\label{prop:abs}
A hybrid formula $f$ is satisfiable iff $$\mathop{min}\limits_{x\in \mathbb{R}^n} F^{abs}_{C,\alpha>0}(x)=0.$$ Moreover if $l\in \mathbb{R}^n$ and $F_{C,\alpha}^{abs}(l)=0$, then \texttt{sgn}$(l)$ satisfies $f$.
\end{proposition}
We refer to $\alpha\cdot \sum_{i\in [n]}(x_i^2-1)^2$ as \textbf{penalty terms} that drag all variables to be either $1$ or $-1$. 

% In the above we proposed a general reduction from hybrid SAT solving to unconstrained continuous optimization. The unconstrained reduction is, however, not valid for MaxSAT. For MaxSAT problem, the global minimum point of the objective value is generally different from the solution to the MaxSAT problem, i.e., $\mathop{argmin}\limits_{x\in \mathbb{R}^n}\sum_{c\in C} |\texttt{FE}_c| + \alpha\cdot \sum_{i\in [n]}(x^2-1)^2$ usually can not be extended into a solution to MaxSAT.
Above we propose general reductions from hybrid SAT to unconstrained continuous optimization. However, the reduction does not extend to MaxSAT problems, i.e., $\mathop{argmin}_{x\in \mathbb{R}^n}F_{C,\alpha}^{sq}$ typically cannot be extended into a valid MaxSAT solution.

\section{The Constraint Penalty Term: When is it Necessary?}

%\textcolor{magenta}{Tasos: I think this is what makes this work different and it is not about algorithms, but analyzing for SAT problems when regularized objectives are necessary and when not.}
Unconstrained formulations proposed in the last section contain penalty terms given by $\sum_{i\in [N]}(x_i^2-1)^2$, which ensures the soundness of the reduction from SAT to unconstrained optimization. During optimization, the dynamics of the penalty terms drag the optimizer toward a Boolean assignment. 
% The unconstrained formulations were proposed in a line of previous works \cite{} for Boolean formulas in CNF, as a special case of hybrid SAT formulas. It is not necessary for the formulations in \cite{} to contain the penalty terms to make the reduction from SAT solving to continuous optimization valid. A natural question is whether the penalty terms are essentially necessary for the validity of the reduction. In this section, we confirm that the penalty terms are necessary for general hybrid constraints. We also give a classification for commonly used types of constraints regarding whether penalty term is needed. 
Similar unconstrained formulations for CNF formulas, as a special case of hybrid formulas, were previously proposed in \cite{gu1999optimizing}. Notably, those formulations for CNF in \cite{gu1999optimizing} do not require penalty terms to establish a valid reduction from SAT to continuous optimization. This naturally raises the question: are penalty terms fundamentally necessary for the soundness of such reductions for general constraints? In this section, we answer the question above and provide a classification of commonly used constraint types based on whether penalty terms are required or can be omitted. We start by formalizing the types of constraints that do not need penalty terms by defining ``rounding friendly".

\begin{definition} (Rounding friendly) A Boolean constraint is called $\epsilon$-Rounding-friendly if there exists $\epsilon\ge 0$, such that $\forall l\in \mathbb{R}^n$, if $|\texttt{FE}_c(l)|\le \epsilon$,  then  $\texttt{sgn}(l)$ satisfies the constraint $c$.
\end{definition}
In other words, if a constraint is rounding friendly, then all real points that evaluate $\texttt{FE}_c$ to sufficiently small magnitude can be rounded to a satisfying assignment. If all constraints are rounding friendly, then the penalty term is not necessary for the soundness of the reduction from SAT to unconstrained optimization, as the following corollary states. 

% \begin{proposition}
% If all constraints of a formula $f$ are rounding-friendly, then
% $f$ is satisfiable iff $min_x F^{sq}_{C, \alpha=0}(x) = 0$, i.e., the penalty term is not necessary for the reduction from SAT to unconstrained optimization.    
% \end{proposition}
% \begin{proof}
% Suppose all constraints in the formula is rounding-friendly, we prove that $f$ is satisfiable iff $min_x F_{sq, \alpha=0}(x) = 0$ in the following.

% $\Rightarrow$: Same with the proof of \ref{prop:sq}.

% $\Leftarrow$: If  $min_x F_{sq, \alpha=0}(x) = 0$, then there exists a real vector $l\in \mathbb{R}^n$ such that $F_{sq, \alpha=0}(x)=0$. Since $F_{sq, \alpha=0}$ is a sum-of-square,  we have $FE_c(l)=0$ for all $c\in C$. Since each $c\in C$ is rounding-friendly, $\texttt{sgn}(l)=0$ and $\texttt{sgn}(l)$ is a Booolean assignment that satisfies all constraints. Hence $f$ is satisfiable.
% \end{proof}

\begin{corollary}
If all constraints of a formula $f$ are $\epsilon$-rounding-friendly for $\epsilon\ge 0$, then
$f$ is satisfiable iff $\mathop{min}\limits_{x\in\mathbb{R}^n} F^{sq/abs}_{C, \alpha=0}(x) = 0$. If $l\in \mathbb{R}^n$ and $F_{C,\alpha=0}^{sq/abs}(l)=0$, then $\texttt{sgn}(l)$ satisfies $f$.
\label{coro:rounding-friendly}
\end{corollary}
% The detailed proof of Proposition~\ref{prop:rounding-friendly} is given in Appendix~\ref{app:proof_rounding_friendly}. 
Thus, to determine whether the penalty terms are necessary in the unconstrained setting of a Boolean formula, it suffices to check if all constraints are rounding-friendly. We investigate whether several types of commonly used hybrid constaints are rounding friendly, including  \texttt{OR} clause, XOR, cardinality constraint, NAE and pseodu-Boolean constraints. 

% \begin{proposition} 
% \label{prop:eg_roundingfirendly}
% XOR, disjunctive clauses (CNF), and not-all-equal (NAE) are $0$-rounding-friendly, while general pseudo-Boolean constraints are not $\epsilon$-rounding-friendly for all $\epsilon\ge 0$. 
% \end{proposition}
% \begin{proof}
%     \textbf{OR clause (CNF)}: the Fourier expansion of a disjunctive clause can be written as $FE_{clause} = 1 + \prod_{l\in L} \frac{(1+l_x)}{2}$. Therefore $FE_{clause}=0$ if and only if at least one literal takes value $-1$. Rounding such a real point gives an assignment with a positive literal, which satisfies the clause.

%     \textbf{XOR}: The Fourier expansion of an XOR constraint is $FE_{XOR}=\frac{1+\prod_{l\in L}l}{2}$. If $FE_{XOR}(a)=0$, then it means that the product of all literals is negative, rounding which gives an assignment satisfying the XOR constraint.

%     \textbf{Pseudo-Boolean constraints}: consider the example of $c=x_1\wedge x_2$ (can be also viewed as a cardinality constraint). (This counterexample needs to be finished)
% \end{proof}
\begin{proposition} 
\label{prop:eg_roundingfirendly}
XOR, \texttt{OR} clauses, and NAE are $0$-rounding-friendly. Cardinality constraints and pseudo-Boolean constraints are not $\epsilon$-rounding-friendly for all $\epsilon\ge 0$. 
\end{proposition}
% The detailed proof of Proposition~\ref{prop:eg_roundingfirendly} is given in Appendix~\ref{app:proof_roundingfriendly_examples}. 
It is desirable to have a classification of all types of constraints regarding rounding friendly. In the following we give a necessary condition rounding friendly constraints.

\begin{definition} (Isolated violations) A Boolean constraint $c$ has isolated violations if there do not exist two Boolean assignments $b_1,b_2$ with hamming distance 1, both violate $c$.
\end{definition}

A constraint with isolated violations can be ``easily fixed" since flipping the value of an arbitrary variable makes the constraint satisfied. It is easy to verify that \texttt{XOR}, \texttt{OR} clauses, and \texttt{Not-All-Equal} have isolated violations. General cardinality constraints and pseudo-Boolean constraints are not with isolated violations.   We show that having isolated violations is a necessary but insufficient condition for rounding friendly.

% \begin{proposition}
% \label{prop:left}
%     If a constraint is rounding-friendly, then it must have isolated violations. 
% \end{proposition}
% \begin{proof}
% First, notice that if a constraint is rounding-friendly for some $\epsilon$, then it must be $0$-rounding-friendly. Equivalently, we show that if a constraint is not singly approvable, then it is not $0$-rounding-friendly. 

%     We assume $c$ has at least two variables. If $c$ is not singly approvable, then there exists a variable, say  $x_1$ and a Boolean assignment $b$ to $x_2,\cdots,x_n$ such that $c|_{X\setminus x_1 \gets b}\equiv \texttt{false}$. Since $c$ is not constant, there must exist a variable $x_k$  with minimum index $k$ such that $c|_{x_{k+1},\cdots,x_{n}\gets b}=l_k\wedge f(x_{1},\cdots, x_{k-1})$. WLOG, assume $b_k=1$, then  $W=c|_{x_{k+1},\cdots,x_{n}\gets b}=x_k\wedge f(x_{1},\cdots, x_{k-1})$. Then $FE_W$ can be written as $FE_W=(x_k-1)\cdot g(x_1,\cdots,x_{k-1})+1$. We also have $g(x_1,\cdots,x_{k-1})\not\equiv 0$, otherwise $k$ is not minimum. Then by Proposition \ref{prop:anyRealNumber}, there exists $a'\in \mathbb{R}^n$ such that $g(a')=2$. Let $x_k=1/2$. Then we have $FE_W(x_k\gets 1/2, a')=-1/2\cdot 2+1=0$. However, $FE_c|_{x_k\gets 1, x_1,\cdots,x_{k-1}\gets \texttt{sgn}(a')}=1$, which means rounding the zero point of $FE_c$ in the real domain does not give a solution, which means that $c$ is not rounding-friendly. 
% \end{proof}
\begin{theorem}
\label{prop:left}
If a constraint is rounding-friendly, then all violations of the constraint are isolated.
\end{theorem}
% The detailed proof of Proposition~\ref{prop:left} is given in Appendix~\ref{app:proof_left}.

Although \texttt{OR} clauses, XOR and NAE are both rounding friendly and with isolated violations, the inverse of Theorem \ref{prop:left} is generally not true. The following example showcases that there can be constraints with isolated violations but not rounding friendly. 
\begin{example}
    Consider the constraint $c$ of length $3$ with violation set $\{\texttt{FFF},\texttt{FTT},\texttt{TFT}\}$ with Fourier Expansion $\texttt{FE}_c=\frac{3}{8}+\frac{1}{8}x_1+\frac{1}{8}x_2-
    \frac{1}{8}x_3+\frac{1}{4}x_1x_2x_3$. All the zeros of $c$ is given by 
    $x_1x_2\neq \frac{1}{2}, x_3=\frac{-x_1-x_2-3}{2x_1x_2-1}.$ Then let $(0.1,0.1,\frac{160}{49})$ is a zero of $\texttt{FE}_c$ but rounding it gives $\texttt{FFF}$, which does not satisfy the constraint.
\end{example}

Next, we also show that \texttt{OR}s, XORs, and NAEs are $\epsilon$-rounding-friendly with positive $\epsilon$.

% \begin{proposition}
% \label{prop:CNFW}
%     An \texttt{OR} clause of length $k$ is  $\frac{1}{2^k}$-rounding-friendly. An XOR constraint of length $k$ is $\frac{1}{2}$-rounding-friendly.
% \end{proposition}

% \begin{proof}
%     If $c$ is a disjunctive clause, then its Fourier expansion is $\prod_x (\frac{1+l}{2})$. If $|\prod_x (\frac{1+l}{2})|< \frac{1}{2^k}$. Then there exists $l_i$ such that $\frac{l+1}{2}<\frac{1}{2}$, i.e., $l<0$. Therefore, rounding the value of $l_i$ gives an assignment that satisfies the constraint.

%     If $c$ is an XOR constraint, then its Fourier Expansion is $\frac{1-\prod_i l_i}{2}$. If $1-\prod_i l_i<1$, then $\prod_i l_i<0$, which indicates that rounding the real assignment gives a solution.
% \end{proof}
\begin{proposition}
\label{prop:CNFW}
An \texttt{OR} clause of length $k$ is  $\frac{1}{2^k}$-rounding-friendly. An XOR constraint is $\frac{1}{2}$-rounding-friendly. An NAE of length $k$ is $\frac{1}{2^{k-1}}$-rounding-friendly.
\end{proposition}
% The detailed proof of Proposition~\ref{prop:CNFW} is given in Appendix~\ref{app:proof_CNFW}.
By Proposition \ref{prop:CNFW}, we can bound the number of violated constraints by the value of the unconstrained formulation in MaxSAT problem. 

\begin{corollary}
\label{coro:bound}
 Suppose $F_{C,\alpha=0}^{sq}(l)=W$ for $l\in \mathbb{R}^n$. If the formula  $f$ is k-CNF , then \texttt{sgn}$(l)$ violates at most $\ceil{4^k\cdot W}$ constraints. If $f$ is pure-\texttt{XOR}, then the rounded assignment violates at most $\ceil{4W}$ constraints. If $f$ is k-\texttt{NAE}, then \texttt{sgn}$(l)$ violates $\ceil{4^{k-1}W}$ constraints.
\end{corollary}

In general, including a penalty term is necessary for soundness in reductions from SAT to unconstrained optimization. However, for certain constraints—such as XOR, OR, or NAE, it can be omitted without sacrificing soundness. Moreover, we show that only constraints with isolated violations can safely omit penalty terms, though the converse does not always hold.

\section{Implementation Details} 
\subsection{Gradient Computation}
In GradSAT \cite{gradsatAAAI}, a follow-up work of FourierSAT with accelerated gradient computation, the gradient of $\texttt{FE}_c$, i.e., $\frac{\partial\texttt{FE}_c}{\partial X}$ of a CNF clause, XOR, NAE and a cardinality constraint with length $k$ can be computed in $O(k^2)$. The gradient of $\texttt{FE}_c^2$ or $|\texttt{FE}_c|$ can be computed by the chain rule as 
$$
\frac{\partial \texttt{FE}_c^2}{\partial x_i}(l) = 2\cdot \texttt{FE}_c(l)\cdot \frac{\partial\texttt{FE}_c}{\partial x_i}(l)
,\,\,\,\frac{\partial|\texttt{FE}_c|}{\partial x_i}(l) = \texttt{sgn}(\texttt{FE}_c(l))\cdot \frac{\partial\texttt{FE}_c}{\partial x_i}(l) (l_i\neq 0),
$$
where \texttt{sgn}$(M)=1$ for $M>0$ and $-1$ otherwise.
Unlike the absolute-value formulation, the square formulation is everywhere differentiable and provides a smoother landscape for optimization.  Moreover, in the square formulation, the gradient contributed by a constraint $c$ is discounted by the value of $\texttt{FE}_c$. Therefore, when the Fourier expansion of a constraint $c$ approximately evaluates to $0$ (intuitively $c$ is close to be satisfied), the gradient provided by $c$ becomes insignificant, which can be viewed as a self-adaptive weighting procedure, similar to Langrange methods~\cite{nocedal2006numerical}.

\subsection{Unconstrained Continuous Optimizers}
In FourierSAT, the continuous optimizers are required to admit box constraints with the form $-1\le x\le 1$. The method in this chapter lifts this requirement and enhances FourierSAT with unconstrained solvers. In the evaluation, we include the following set of continuous optimizers: 
\begin{itemize}
    \item  (Projected) gradient descent (PGD, constrained/GD, unconstrained). In PGD, if a point $l\in\mathbb{R}^n$ is reached outside $[-1,1]^n$, then the search returns to the point in $[-1,1]^n$ with minimum distance to $l$. 
    \item Sequential least squares programming, SLSQP (Constrained/Unconstrained), a second-order method that was shown to have the best performance compared to BFGS and CG in \cite{pellow2021comparison}. 
    \item Adam (Unconstrained) \cite{kingma2014adam}. One of the most popular optimizers used in neural network training. As a gradient-based approach, Adam integrates the momentum-based techniques~\cite{zou2019sufficient} techniques to accelerate the convergence and escape from local optimum.
    \item Hamilton-Jacobi-based Proximal Point (HJ-PROX)~\cite{osher2023hamilton}. A proximal-point method where the proximal operator is approximated using Laplace's Method~\cite{tibshirani2024laplace, meng2025recent, heaton2024global}. 
\end{itemize}

% \begin{itemize}
%     \item (Projected) gradient descent (PGD, constrained/GD, unconstrained). In PGD, if a point $l\in\mathbb{R}^n$ is reached outside $[-1,1]^n$, then the search returns to the point in $[-1,1]^n$ with minimum distance to $l$. 
%     \item SLSQP (Constrained/Unconstrained), a second-order method that was shown to have the best performance compared to BFGS and CG in \cite{pellow2021comparison}. 
%     \item Adam (Unconstrained) \cite{kingma2014adam}. One of the most popular optimizers used in neural network training. 
%     As a gradient-based approach, Adam integrates the momentum-based techniques~\cite{zou2019sufficient} techniques to accelerate the convergence and escape from local optimum. 
%     \item Hamilton-Jacobi-based Proximal Point (HJ-PROX)~\cite{osher2023hamilton}. A proximal-point method where the proximal operator is approximated using Laplace's Method~\cite{tibshirani2024laplace, meng2025recent, heaton2024global}. 
%     % HJ-PROX was found useful for finding global minimizers of problems in low/moderate dimensions~\cite{heaton2024global}.
%    % \item RMSPROP [U]
%    % \item CG [C/U]
% \end{itemize}

\section{Experiment Results}
We implement the approaches in this chapter on the code-base of FourierSAT \footnote{Available at \url{https://github.com/zzwonder/FourierSAT}.} and conduct evaluations to answer the following research questions. 

\noindent\textbf{RQ1}: How does the penalty coefficient ($\alpha$) affect the performance for unconstrained optimization? 

\noindent\textbf{RQ2}: How do box constraints affect the performance? 

\noindent\textbf{RQ3}: What benefits can unconstrained optimization formulation bring us for solving hybrid SAT problems?

Each experiment was run on an exclusive node in a Linux cluster with 16 processors at 2.63 GHz and 4 GB of RAM per processor. We implement GD/PGD with step size $10^{-3}$. 
We use the implementation of Adam from \cite{kingma2014adam}. The implementation of HJ-PROX is from \cite{osher2023hamilton}. For SLSQP, we use the implementation from the Scipy package~\cite{virtanen2020scipy}. The time limit for each formula is set to 300 seconds.

\subsection{Benchmarks} Let $n$ be the number of variables, $m$ be the number of constraints. The benchmarks consist of the following categories.

\noindent\textbf{Random 3-CNF} $n=1000$, 100 random formulas for each $\frac{m}{n} = \{1.0,1.1, \cdots, 3.6\}$.%It is known that for large $n$, there exists a threshold of $\alpha_0$ such that when $\alpha<\alpha_0$, the formula is w.h.p satisfiable and when $\alpha > \alpha_0$, the formula is w.h.p. unsatisfiable. Experiments indicate that the threshold is near $\alpha= 4.26$. CDCL SAT solvers are known to perform badly on those instances while local search solvers can solve them up to a high density. %Survey propagation is the only known SAT solver which can solve instances with large number of variables and  density very close to 4.26.

\noindent\textbf{Random 2-XOR} $n=1000$, 100 random formulas for each $\frac{m}{n} = \{0.1,0.2,\cdots,1.0\}$. %Unlike CNF problems, the satisfiability problem of XORs is a problem in PTIME, solvable by Gaussian Elimination. Random 2-XOR instances are known to have satisfiability threshold of density at $1.0$. However, SAT solvers are known to perform poorly on XOR constraints [citation].

\noindent\textbf{Random Cardinality constraints} For $n\in\{50,100,150\}$, $r_P\in\{0.5,0.6,0.7\}$, $r_V\in\{0.2,0.3,0.4,0.5\}$, $100$ formulas consisting of $r_P\cdot n$ cardinality constraints where each cardinality constraint contains  $r_V\cdot n$ variables randomly sampled from $\{x_1,\cdots,x_n\}$ are generated. The comparator is uniformly sampled from $\{\ge, \le\}$ and the right-hand side is set to $\frac{r_V\cdot n}{2}$.

\subsection{Results and Analysis}
\subsubsection{RQ1: Effect of penalty terms} 
We evaluate the effect of the penalty terms for random 3-CNF and 2-XOR constraints, which in theory do not need penalty terms for soundness, and cardinality (CARD) constraints, which generally need penalty terms. Results are shown in Figure~\ref{fig:cnf_barrier}.

% Although the penalty terms are not necessary for the SAT problem for CNF and XOR, we study their effect on the performance of the algorithms. 
% Let the coefficient of the penalty be $\{0,0.4,0.8,1.2,1.6\}$. 
% The following figure shows the results on random CNF and XOR constraints on a variety of solvers.

%   \begin{figure}[!ht]
%     \centering
    
%     \subfigure[SLSQP-SQUARE]{
%     \includegraphics[width=.35\linewidth]{cp_format/figures/CNF_1000_penaltyTermSQUARE.png}
%     }
%       \subfigure[SLSQP-ABS]{
%     \includegraphics[width=.35\linewidth]{cp_format/figures/CNF_1000_penaltyTermABS.png}
%     }
%         \subfigure[SLSQP-SQUARE-C]{
%     \includegraphics[width=.35\linewidth]{cp_format/figures/CNF_1000_penaltyTermSQUARE-C.png}
%     }
%     \subfigure[SLSQP-ABS-C]{
% \includegraphics[width=.35\linewidth]{CNF_1000_penaltyTermABS-C.png}
%   }      
%     \caption{Effect of the penalty coefficient in \{0,0.2,0.4,0.6,0.8\} on random CNF instances}
%     \label{fig:cnf_barrier}
%     \end{figure}
\begin{figure}[!ht]
  \hspace*{-1cm}
    \begin{tabular}{ccc}
        (A) SLSQP-SQ-rand3CNF & (B) SLSQP-SQ-rand2XOR & (C) SQ-CARD
        \\
        \includegraphics[width=.35\linewidth]{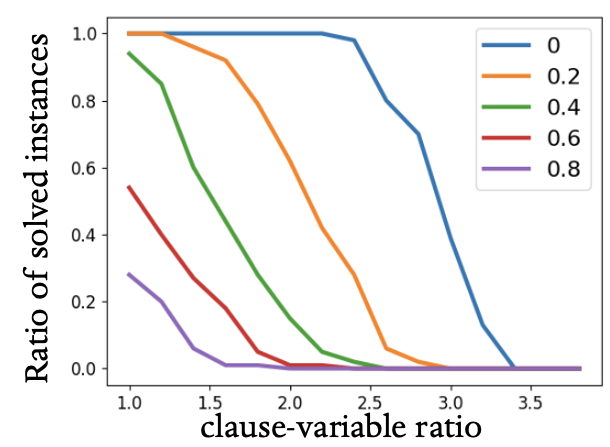}
        & 
        \includegraphics[width=.36\linewidth]{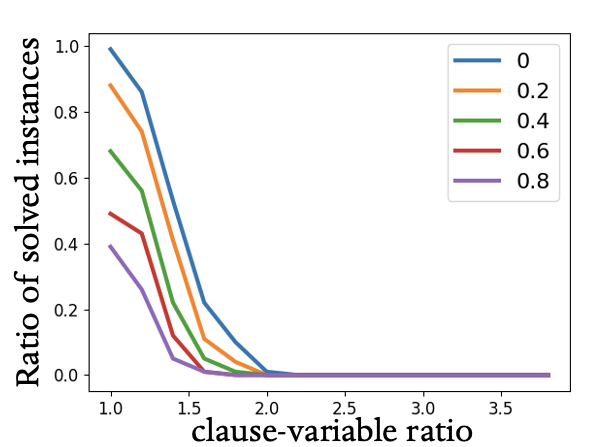}
        & 
        \includegraphics[width=.35\linewidth]{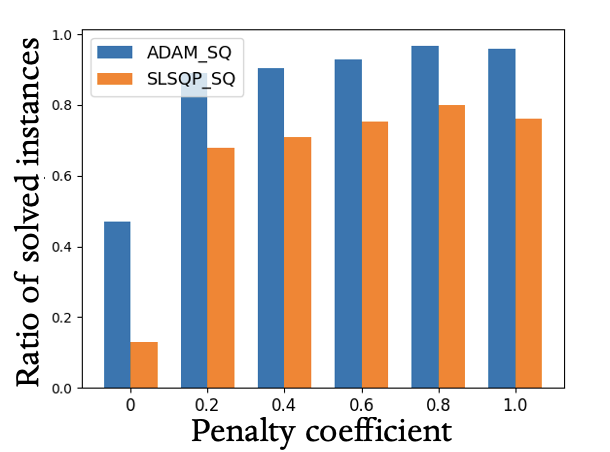}
        \\
    \end{tabular}
    \caption{\small{Effect of the penalty coefficient $\alpha\in\{0,0.2,0.4,0.6,0.8\}$ on (A) random 3-CNF on the square formulation and optimizer SLSQP. (B) random 2-XOR on the square formulation and SLSQP. (C) random CARD on the square formulation and optimizer SLSQP and Adam.}}
    \label{fig:cnf_barrier}
\end{figure}
% We observe that the penalty terms generally have a negative effect on FourierSAT on CNF and XOR formulas on the first two. 
% We also evaluate the effect on CARD constraints, where in theory the penalty term is needed. The performance can be viewed in Figure \ref{fig:cnf_barrier} (right). It can be viewed that with a non-zero penalty coefficient, the solvers are able to solve significantly more instances than the case with no penalty terms, which aligns with our theoretical discovery. 

We observe that a non-zero penalty coefficient degrades solver performance on random 3-CNF and 2-XOR constraints, which theoretically do not require penalty terms (Figure 1A, B). The larger the coefficient, the worse the performance, and this effect appears in both constrained and unconstrained settings. Conversely, for constraints that do require penalty terms (such as cardinality constraints), introducing a non-zero penalty coefficient significantly improves solver performance (Figure 1C).

\subsubsection{RQ2: Effect of different formulations and box constraints} 
We evaluate different formulations (linear/square/absolute value, constrained/unconstrained) and compare their performance on the same optimizer PGD/GD. The results for random 3-CNF and random 2-XOR problems are presented in Figure \ref{fig:formulations_CNFXOR}A,B. 

We observe that the continuous formulation is critical to the performance. For both random CNF and XOR problems, the square formulation outperforms absolute-value and linear for both optimizers (GD, SLSQP), which could be possibly explained by the smoothness and auto-adaptive property of the square formulation. Allowing an unconstrained setting improves the performance of GD on square formulation in both random CNF and XOR benchmarks, indicating that allowing ``moving out of the box" could be beneficial. 

%   \begin{figure}[!ht]
%     \centering
%         \subfigure[GD on random 3-CNF]{
%     \includegraphics[width=.4\linewidth]{cp_format/figures/CNF_1000_formulation_GD.png}
%     }
%     \subfigure[SLSQP on random 3-CNF]{
% \includegraphics[width=.4\linewidth]{cp_format/figures/CNF_1000_formulation_SLSQP.png}
%   }  
%    \subfigure[GD on random 2-XOR]{
%     \includegraphics[width=.4\linewidth]{cp_format/figures/XOR_1000_formulation_GD.png}
%     }
%     \subfigure[SLSQP on random 2-XOR]{
% \includegraphics[width=.4\linewidth]{cp_format/figures/XOR_1000_formulation_SLSQP.png}
%   }  
%   \caption{Results of GD and SLSQP on random 3-CNF and 2-XOR problems with different formulations. The SQUARE formulation outperformed the other two formulations (ABS, LINEAR) for both optimizers (GD, SLSQP). Allowing unconstrained setting improves the performance of GD and improves the virtual best solver of those two solvers.}
%   \label{fig:formulations_CNFXOR}
%     \end{figure}
\begin{figure}[!ht]
    \centering
      \hspace*{-1cm}
    \begin{tabular}{ccc}
        (A) GD-rand3CNF & (B) GD-rand2XOR & (C) multiple-solvers-rand3CNF
        \\
        \includegraphics[width=.33\linewidth]{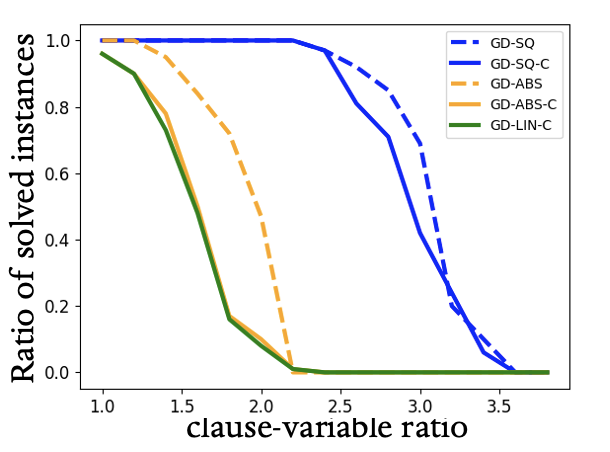} \label{subfig:form1}
        & 
        \includegraphics[width=.33\linewidth]{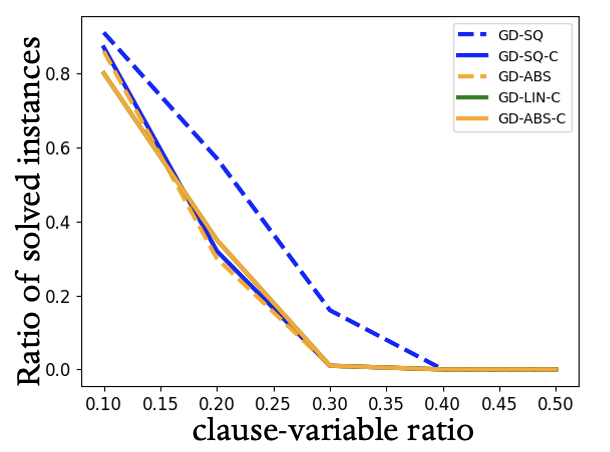} \label{subfig:form2}
        &
        \includegraphics[width=.33\linewidth]{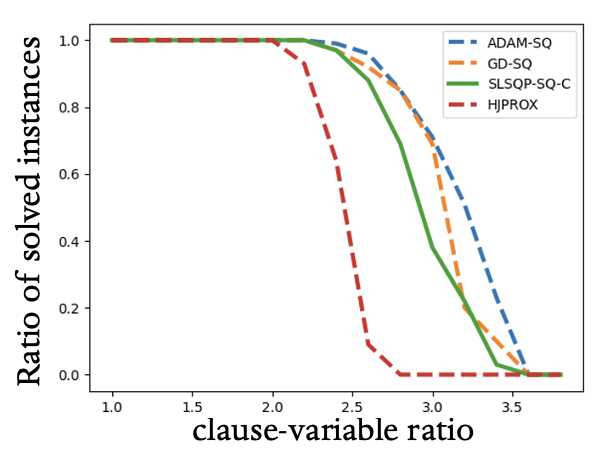} \label{subfig:form2}
    \end{tabular}
    \caption{\small{(A) Performance of different formulations on GD/PGD in random 3CNF formulas.         ``-c" indicates constrained formulation. (B) Performance of different formulations on GD/PGD in random 2XOR formulas. The square formulation outperformed the other two formulations (abs, linear). Allowing an unconstrained setting improves the performance of GD on square formulation in both random CNF and XOR benchmarks. (C) Performance on multiple solvers in their best formulation on random 3CNF. Adam improves the portfolio of the previous set of continuous optimizers.}}
    \label{fig:formulations_CNFXOR}
\end{figure}

 \subsubsection{RQ3: Benefit of allowing unconstrained continuous optimizers}
To see whether accepting unconstrained optimizers enhances the solver portfolio, we compare the performance of different optimizers on their best formulations on random 3-CNF formulas. The results are shown in Figure \ref{fig:formulations_CNFXOR}C. Adam on the square formulation is able to solve random 3-CNF with higher clause-variable ratio than all solvers previously used in FourierSAT, showcasing the potential of using machine-learning-inspired optimizer in SAT solving.

% \paragraph{Summary}

% \section{Conclusion and Future Directions}
% In this paper, we investigate the potential of unconstrained optimization for solving Boolean SAT problems. By applying non-negative operators on Fourier expansions, the box constraints can be removed and more types of constraints can be applied. Therefore, a comprehensive set of continuous optimizers can be integrated into FourierSAT. 
% The side effect of removing the box constraints is that penalty terms must be added to maintain the theoretical correctness. We also investigate the conditions where the penalty term is necessary. In practice, the penalty terms usually yields a negative effect. 
% Through experimental studies, we find that the formulation of the optimization problem is very important. 
% By allowing more types of optimizers, some optimizers are shown to have significantly better performance than constrained optimizers. We note that our framework can continue benefiting from the advancement of general continuous optimizers and hardware such as GPU \cite{cen2023massively} and learned optimization algorithms for improved speedups~\cite{chen2022learning, li2017learning, heaton2023explainable, mckenzie2024three}.

% For future directions, one is to try other formulations inspired from other areas, e.g., cross entropy. The main challenge is how to alleviate numerical issues.
    
\section{Conclusion and Future Directions}
In this paper, we explore unconstrained optimization for Boolean SAT problems by applying non-negative operators on Fourier expansions, thus eliminating box constraints and enabling integration of diverse continuous optimizers into FourierSAT. 
% While penalty terms must be added to maintain theoretical correctness, our analysis shows they often negatively impact practical performance. 
Our experiments demonstrate that optimization problem formulation is crucial, with certain unconstrained optimizers significantly outperforming constrained ones. This framework is positioned to benefit from advances in continuous optimization, specialized hardware \cite{cen2023massively}, and learned optimization algorithms \cite{chen2022learning, li2017learning, heaton2023explainable, mckenzie2024three}. Future work can explore alternative formulations from other domains like cross-entropy methods, with the main challenge being the mitigation of numerical stability issues.

%%
%% Bibliography
%%

%% Please use bibtex, 
\section*{Acknowledgements}
Work supported in part by NSF grants IIS-1527668, CCF-1704883, IIS-1830549, DoD MURI grant N00014-20-1-2787, Andrew Ladd Graduate Fellowship of Rice Ken Kennedy Institute, and an award from the Maryland Procurement Office. 

% \bibliography{lipics-v2021-sample-article}
\bibliography{ref.bib}

\begin{thebibliography}{10}

\bibitem{glucose}
Gilles Audemard and Laurent Simon.
\newblock {Lazy Clause Exchange Policy for Parallel SAT Solvers}.
\newblock In {\em SAT 2014}, pages 197--205, Cham, 2014.

\bibitem{sparrow}
Adrian Balint and Andreas Fr{\"o}hlich.
\newblock Improving {S}tochastic {L}ocal {S}earch for {SAT} with a {N}ew {P}robability {D}istribution.
\newblock In {\em SAT 2010}, pages 10--15, 2010.

\bibitem{probSAT}
Adrian Balint and Uwe Sch{\"o}ning.
\newblock {Choosing Probability Distributions for Stochastic Local Search and the Role of Make versus Break}.
\newblock In {\em SAT 2012}, pages 16--29, 2012.

\bibitem{bayardo1997using}
Roberto~J Bayardo~Jr and Robert Schrag.
\newblock Using csp look-back techniques to solve real-world sat instances.
\newblock In {\em Aaai/iaai}, pages 203--208. Citeseer, 1997.

\bibitem{bertsekas1982projected}
Dimitri~P Bertsekas.
\newblock Projected newton methods for optimization problems with simple constraints.
\newblock {\em SIAM Journal on control and Optimization}, 20(2):221--246, 1982.

\bibitem{Biclique-Cryptanalysis-of-the-Full-AES}
Andrey Bogdanov, Dmitry Khovratovich, and Christian Rechberger.
\newblock {Biclique Cryptanalysis of the Full AES}.
\newblock In {\em ASIACRYPT 2011}, 2011.

\bibitem{braunstein2005survey}
Alfredo Braunstein, Marc M{\'e}zard, and Riccardo Zecchina.
\newblock Survey propagation: An algorithm for satisfiability.
\newblock {\em Random Structures \& Algorithms}, 27(2):201--226, 2005.

\bibitem{Walksat-implementation}
Shaowei Cai, Kaile Su, and Chuan Luo.
\newblock Improving {W}alksat for {R}andom k-{S}atisfiability {P}roblem with k $>$ 3.
\newblock In {\em AAAI}, 2013.

\bibitem{cen2023massively}
Yunuo Cen, Zhiwei Zhang, and Xuanyao Fong.
\newblock Massively parallel continuous local search for hybrid sat solving on gpus.
\newblock {\em arXiv preprint arXiv:2308.15020}, 2023.

\bibitem{SATprobabilisticInference}
Mark Chavira and Adnan Darwiche.
\newblock On probabilistic inference by weighted model counting.
\newblock {\em Artificial Intelligence}, 172(6):772 -- 799, 2008.
\newblock URL: \url{http://www.sciencedirect.com/science/article/pii/S0004370207001889}, \href {https://doi.org/10.1016/j.artint.2007.11.002} {\path{doi:10.1016/j.artint.2007.11.002}}.

\bibitem{chen2022learning}
Tianlong Chen, Xiaohan Chen, Wuyang Chen, Howard Heaton, Jialin Liu, Zhangyang Wang, and Wotao Yin.
\newblock Learning to optimize: A primer and a benchmark.
\newblock {\em Journal of Machine Learning Research}, 23(189):1--59, 2022.

\bibitem{chowdhury2024tassat}
Md~Solimul Chowdhury, Cayden~R Codel, and Marijn~JH Heule.
\newblock Tassat: Transfer and share sat.
\newblock In {\em International Conference on Tools and Algorithms for the Construction and Analysis of Systems}, pages 34--42. Springer, 2024.

\bibitem{Graph-coloring-with-cardinality-constraints}
M.-C. Costa, D.~de~Werra, C.~Picouleau, and B.~Ries.
\newblock {Graph Coloring with Cardinality Constraints on the Neighborhoods}.
\newblock {\em Discrete Optimization}, 6(4):362 -- 369, 2009.
\newblock URL: \url{http://www.sciencedirect.com/science/article/pii/S1572528609000231}, \href {https://doi.org/10.1016/j.disopt.2009.04.005} {\path{doi:10.1016/j.disopt.2009.04.005}}.

\bibitem{davis1962machine}
Martin Davis, George Logemann, and Donald Loveland.
\newblock {A Machine Program for Theorem-Proving}.
\newblock {\em Communications of the ACM}, 5(7):394--397, 1962.

\bibitem{nae-coloring}
Irit Dinur, Oded Regev, and Clifford Smyth.
\newblock {The Hardness of 3-Uniform Hypergraph Coloring}.
\newblock {\em Combinatorica}, 25(5):519--535, Sep 2005.
\newblock \href {https://doi.org/10.1007/s00493-005-0032-4} {\path{doi:10.1007/s00493-005-0032-4}}.

\bibitem{minisat}
Niklas E{\'e}n and Niklas S{\"o}rensson.
\newblock {An Extensible SAT-Solver}.
\newblock In {\em SAT}, pages 502--518, 2003.

\bibitem{fleury2020cadical}
ABKFM Fleury and Maximilian Heisinger.
\newblock Cadical, kissat, paracooba, plingeling and treengeling entering the sat competition 2020.
\newblock {\em Sat Competition}, 2020:50, 2020.

\bibitem{gafni1984two}
Eli~M Gafni and Dimitri~P Bertsekas.
\newblock Two-metric projection methods for constrained optimization.
\newblock {\em SIAM Journal on Control and Optimization}, 22(6):936--964, 1984.

\bibitem{gu1999optimizing}
Jun Gu, Qianping Gu, and Dingzhu Du.
\newblock On optimizing the satisfiability (sat) problem.
\newblock {\em Journal of Computer Science and Technology}, 14(1):1--17, 1999.

\bibitem{cnfencodingofXOR}
Matthew Gwynne and Oliver Kullmann.
\newblock On sat representations of xor constraints.
\newblock In Adrian-Horia Dediu, Carlos Mart{\'i}n-Vide, Jos{\'e}-Luis Sierra-Rodr{\'i}guez, and Bianca Truthe, editors, {\em Language and Automata Theory and Applications}, pages 409--420, Cham, 2014. Springer International Publishing.

\bibitem{heaton2023explainable}
Howard Heaton and Samy~Wu Fung.
\newblock Explainable ai via learning to optimize.
\newblock {\em Scientific Reports}, 13(1):10103, 2023.

\bibitem{heaton2024global}
Howard Heaton, Samy Wu~Fung, and Stanley Osher.
\newblock Global solutions to nonconvex problems by evolution of {Hamilton-Jacobi} pdes.
\newblock {\em Communications on Applied Mathematics and Computation}, 6(2):790--810, 2024.

\bibitem{heule2016solving}
Marijn~JH Heule, Oliver Kullmann, and Victor~W Marek.
\newblock Solving and verifying the boolean pythagorean triples problem via cube-and-conquer.
\newblock In {\em International Conference on Theory and Applications of Satisfiability Testing}, pages 228--245. Springer, 2016.

\bibitem{hoos2000local}
Holger~H Hoos and Thomas St{\"u}tzle.
\newblock Local search algorithms for sat: An empirical evaluation.
\newblock {\em Journal of Automated Reasoning}, 24(4):421--481, 2000.

\bibitem{horbach2010boolean}
Andrei Horbach.
\newblock A boolean satisfiability approach to the resource-constrained project scheduling problem.
\newblock {\em Annals of Operations Research}, 181:89--107, 2010.

\bibitem{kader2017novel}
ASM~Abdul Kader and Mikhail Dorojevets.
\newblock Novel integration of dimetheus and walksat solvers for k-sat filter construction.
\newblock In {\em 2017 IEEE Long Island Systems, Applications and Technology Conference (LISAT)}, pages 1--5. IEEE, 2017.

\bibitem{kan2021pnkh}
Kelvin Kan, Samy~Wu Fung, and Lars Ruthotto.
\newblock Pnkh-b: A projected newton--krylov method for large-scale bound-constrained optimization.
\newblock {\em SIAM Journal on Scientific Computing}, 43(5):S704--S726, 2021.

\bibitem{katz2017reluplex}
Guy Katz, Clark Barrett, David~L Dill, Kyle Julian, and Mykel~J Kochenderfer.
\newblock Reluplex: An efficient smt solver for verifying deep neural networks.
\newblock In {\em Computer Aided Verification: 29th International Conference, CAV 2017, Heidelberg, Germany, July 24-28, 2017, Proceedings, Part I 30}, pages 97--117. Springer, 2017.

\bibitem{kingma2014adam}
Diederik~P Kingma and Jimmy Ba.
\newblock Adam: A method for stochastic optimization.
\newblock {\em arXiv preprint arXiv:1412.6980}, 2014.

\bibitem{kyrillidis2021solving}
Anastasios Kyrillidis, Anshumali Shrivastava, Moshe~Y Vardi, and Zhiwei \textbf{Zhang}.
\newblock Solving hybrid boolean constraints in continuous space via multilinear fourier expansions.
\newblock {\em Artificial Intelligence}, page 103559, 2021.

\bibitem{fouriersat}
Anastasios {Kyrillidis}, Anshumali {Shrivastava}, Moshe~Y. {Vardi}, and Zhiwei {Zhang}.
\newblock {FourierSAT: A Fourier Expansion-Based Algebraic Framework for Solving Hybrid Boolean Constraints}.
\newblock In {\em AAAI'20}, 2020.

\bibitem{gradsatAAAI}
Anastasios {Kyrillidis}, Moshe~Y. {Vardi}, and Zhiwei {Zhang}.
\newblock {On Continuous Local BDD-Based Search for Hybrid SAT Solving}.
\newblock In {\em AAAI'21}, 2021.

\bibitem{li2017learning}
Ke~Li and Jitendra Malik.
\newblock Learning to optimize neural nets.
\newblock {\em arXiv preprint arXiv:1703.00441}, 2017.

\bibitem{maplesat}
Jia~Hui Liang.
\newblock {\em Machine Learning for SAT Solvers}.
\newblock PhD thesis, University of Waterloo, December 2018.

\bibitem{marques1999grasp}
Joao~P Marques-Silva and Karem~A Sakallah.
\newblock {{GRASP: A} Search Algorithm For Propositional Satisfiability}.
\newblock {\em IEEE Transactions on Computers}, 48(5):506--521, 1999.

\bibitem{mckenzie2024three}
D~McKenzie, H~Heaton, Q~Li, S~Wu~Fung, S~Osher, and W~Yin.
\newblock Three-operator splitting for learning to predict equilibria in convex games.
\newblock {\em SIAM Journal on Mathematics of Data Science}, 6(3):627--648, 2024.

\bibitem{meng2025recent}
Tingwei Meng, Siting Liu, Samy~Wu Fung, and Stanley Osher.
\newblock Recent advances in numerical solutions for hamilton-jacobi pdes.
\newblock {\em arXiv preprint arXiv:2502.20833}, 2025.

\bibitem{chaff_paper}
Matthew~W. Moskewicz, Conor~F. Madigan, Ying Zhao, Lintao Zhang, and Sharad Malik.
\newblock Chaff: Engineering an efficient sat solver, 2001.
\newblock URL: \url{http://doi.acm.org/10.1145/378239.379017}, \href {https://doi.org/10.1145/378239.379017} {\path{doi:10.1145/378239.379017}}.

\bibitem{nikolic2022survey}
Goran~S Nikoli{\'c}, Bojan~R Dimitrijevi{\'c}, Tatjana~R Nikoli{\'c}, and Mile~K Stojcev.
\newblock A survey of three types of processing units: Cpu, gpu and tpu.
\newblock In {\em 2022 57th International Scientific Conference on Information, Communication and Energy Systems and Technologies (ICEST)}, pages 1--6. IEEE, 2022.

\bibitem{nocedal2006numerical}
Jorge Nocedal and Stephen Wright.
\newblock {\em Numerical optimization}.
\newblock Springer Science \& Business Media, 2006.

\bibitem{O'Donnell:2014:ABF:2683783}
Ryan O'Donnell.
\newblock {\em Analysis of Boolean Functions}.
\newblock Cambridge University Press, New York, NY, USA, 2014.

\bibitem{osher2023hamilton}
Stanley Osher, Howard Heaton, and Samy Wu~Fung.
\newblock A {Hamilton--Jacobi}-based proximal operator.
\newblock {\em Proceedings of the National Academy of Sciences}, 120(14):e2220469120, 2023.

\bibitem{paredes2019principled}
Roger Paredes, Leonardo Due{\~n}as-Osorio, Kuldeep~S Meel, and Moshe~Y Vardi.
\newblock Principled network reliability approximation: A counting-based approach.
\newblock {\em Reliability Engineering \& System Safety}, 191:106472, 2019.

\bibitem{pellow2021comparison}
Aidan Pellow-Jarman, Ilya Sinayskiy, Anban Pillay, and Francesco Petruccione.
\newblock A comparison of various classical optimizers for a variational quantum linear solver.
\newblock {\em Quantum Information Processing}, 20(6):202, 2021.

\bibitem{pblib}
Tobias Philipp and Peter Steinke.
\newblock Pblib -- a library for encoding pseudo-boolean constraints into cnf.
\newblock In Marijn Heule and Sean Weaver, editors, {\em Theory and Applications of Satisfiability Testing -- SAT 2015}, pages 9--16, Cham, 2015. Springer International Publishing.

\bibitem{encoding-handbook-of-satisfiability}
S.~Prestwich.
\newblock {CNF Encodings, Handbook of Satisfiability: Volume 185 Frontiers in Artificial Intelligence and Applications}, 2009.

\bibitem{prestwich2009cnf}
Steven~D Prestwich.
\newblock Cnf encodings.
\newblock {\em Handbook of satisfiability}, 185:75--97, 2009.

\bibitem{schuler2005algorithm}
Rainer Schuler.
\newblock An algorithm for the satisfiability problem of formulas in conjunctive normal form.
\newblock {\em Journal of Algorithms}, 54(1):40--44, 2005.

\bibitem{walksat}
Bart Selman, Henry Kautz, and Bram Cohen.
\newblock {Local Search Strategies for Satisfiability Testing}.
\newblock {\em Second DIMACS Implementation Challenge}, 26, 09 1999.

\bibitem{GSAT}
Bart Selman, Hector Levesque, and David Mitchell.
\newblock {A New Method for Solving Hard Satisfiability Problems}, 1992.
\newblock URL: \url{http://dl.acm.org/citation.cfm?id=1867135.1867203}.

\bibitem{silva1996grasp}
JP~Marques Silva and Karem~A Sakallah.
\newblock Grasp-a new search algorithm for satisfiability.
\newblock In {\em Proceedings of International Conference on Computer Aided Design}, pages 220--227. IEEE, 1996.

\bibitem{soeken2020boolean}
Mathias Soeken, Giulia Meuli, Bruno Schmitt, Fereshte Mozafari, Heinz Riener, and Giovanni De~Micheli.
\newblock Boolean satisfiability in quantum compilation.
\newblock {\em Philosophical Transactions of the Royal Society A}, 378(2164):20190161, 2020.

\bibitem{cmspaper}
Mate Soos, Karsten Nohl, and Claude Castelluccia.
\newblock {Extending {SAT} Solvers to Cryptographic Problems}.
\newblock In {\em {SAT}}, pages 244--257, 2009.

\bibitem{tibshirani2024laplace}
Ryan~J Tibshirani, Samy Wu~Fung, Howard Heaton, and Stanley Osher.
\newblock {L}aplace meets {M}oreau: Smooth approximation to infimal convolutions using {L}aplace's method.
\newblock {\em arXiv preprint arXiv:2406.02003}, 2024.

\bibitem{Vardi14a}
Moshe~Y. Vardi.
\newblock Boolean satisfiability: theory and engineering.
\newblock {\em Commun. {ACM}}, 57(3):5, 2014.

\bibitem{vardi2023solving}
Moshe~Y Vardi and Zhiwei Zhang.
\newblock Solving quantum-inspired perfect matching problems via tutte's theorem-based hybrid boolean constraints.
\newblock {\em arXiv preprint arXiv:2301.09833}, 2023.

\bibitem{SATSolvingVerification}
M.~N. {Velev}.
\newblock Efficient translation of boolean formulas to cnf in formal verification of microprocessors.
\newblock In {\em ASP-DAC 2004: Asia and South Pacific Design Automation Conference 2004 (IEEE Cat. No.04EX753)}, pages 310--315, 2004.

\bibitem{virtanen2020scipy}
Pauli Virtanen, Ralf Gommers, Travis~E Oliphant, Matt Haberland, Tyler Reddy, David Cournapeau, Evgeni Burovski, Pearu Peterson, Warren Weckesser, Jonathan Bright, et~al.
\newblock Scipy 1.0: fundamental algorithms for scientific computing in python.
\newblock {\em Nature methods}, 17(3):261--272, 2020.

\bibitem{yang2021engineering}
Jiong Yang and Kuldeep~S Meel.
\newblock Engineering an efficient pb-xor solver.
\newblock In {\em 27th International Conference on Principles and Practice of Constraint Programming (CP 2021)}. Schloss-Dagstuhl-Leibniz Zentrum f{\"u}r Informatik, 2021.

\bibitem{msthesis}
Zhiwei Zhang.
\newblock Solving hybrid boolean constraints by fourier expansions and continuous optimization.
\newblock Master's thesis, Rice University, January 2020.

\bibitem{zou2019sufficient}
Fangyu Zou, Li~Shen, Zequn Jie, Weizhong Zhang, and Wei Liu.
\newblock A sufficient condition for convergences of adam and rmsprop.
\newblock In {\em Proceedings of the IEEE/CVF Conference on computer vision and pattern recognition}, pages 11127--11135, 2019.

\end{thebibliography}

\appendix

\section{Proofs}\label{sec:proofs}
We restate each of the statements to be proven for ease of readability.

\subsection{Proof of Proposition~\ref{prop:sq}}\label{app:proof_sq}
% \begin{proposition}[Restatement of Proposition~\ref{prop:sq}]
\noindent\textbf{Proposition~\ref{prop:sq} (restated).}
\textit{A hybrid formula $f$ is satisfiable iff $\mathop{min}\limits_{x\in \mathbb{R}^n}F^{sq}_{C,\alpha}(x)=0$ for $\alpha>0$. Moreover if $l\in \mathbb{R}^n$ and $F_{C,\alpha}^{sq}(l)=0$, then \texttt{sgn}$(l)$ satisfies $f$.}

\begin{proof}
    $\Rightarrow$: Suppose $f$ is satisfiable. Then there exists an assignment $b\in\{-1,1\}^n$ such that $\texttt{FE}_c(b)=0$ for all $c\in C$. Therefore, $F_{sq,\alpha}(b)=0$. Since $F^{sq}_{C,\alpha}(x)\ge 0$ for all $x\in\mathbb{R}^n$, it follows that $\mathop{min}\limits_{x\in \mathbb{R}^n}F_{C,\alpha}^{sq}(x)=0$.

 \noindent   $\Leftarrow$: Suppose $\mathop{min}\limits_{x\in \mathbb{R}^n}F^{sq}_{C,\alpha}(x)=0$. Then there exists a real vector $l\in \mathbb{R}^n$ such that $F^{sq}_{C,\alpha}(l)=0$. Since $F^{sq}_{C,\alpha}$ is a sum-of-squares formulation, we have $\texttt{FE}_c(l)=0$ for all $c\in C$, and $(l_i^2-1)=0$ for all $i\in [n]$. Therefore, $l\in\{-1,1\}^n$ is a Boolean assignment. By the definition of Fourier expansion of Boolean functions, the assignment $l$ satisfies the formula $f$, thus proving that $f$ is satisfiable.
\end{proof}

\subsection{Proof of Corollary~\ref{coro:rounding-friendly}}\label{app:proof_rounding_friendly}

\noindent\textbf{Corollary~\ref{coro:rounding-friendly} (restated).}
\textit{If all constraints of a formula $f$ are rounding-friendly, then
$f$ is satisfiable iff $\mathop{min}\limits_{x} F^{sq}_{C, \alpha=0}(x) = 0$, i.e., the penalty term is not necessary for the reduction from SAT to unconstrained optimization.}

\begin{proof}
Suppose all constraints in the formula are rounding-friendly. We prove that $f$ is satisfiable iff $\mathop{min}\limits_{x} F_{sq, \alpha=0}(x) = 0$ as follows:

$\Rightarrow$: Same argument as in the proof of Proposition~\ref{prop:sq}.

$\Leftarrow$: If $\mathop{min}\limits_{x} F_{sq, \alpha=0}(x) = 0$, there exists a real vector $l\in \mathbb{R}^n$ such that $F_{sq, \alpha=0}(l)=0$. Since $F_{sq, \alpha=0}$ is a sum-of-squares, we have $FE_c(l)=0$ for all $c\in C$. Since each $c\in C$ is rounding-friendly, it follows that $\texttt{sgn}(l)$ is well-defined and is a Boolean assignment that satisfies all constraints. Hence, $f$ is satisfiable.
\end{proof}

\section{Proof of Proposition~\ref{prop:eg_roundingfirendly}}\label{app:proof_roundingfriendly_examples}

\noindent\textbf{Proposition~\ref{prop:eg_roundingfirendly} (restated).}
\textit{XOR, disjunctive clauses (CNF), and not-all-equal (NAE) are $0$-rounding-friendly, while general pseudo-Boolean constraints are not $\epsilon$-rounding-friendly for all $\epsilon\ge 0$.}

\begin{proof}
Without loss of generality, we assume all constraints do not contain negative literals. This is because the negation of a Boolean variable of a constraint is equivalent with negating the value of that variable in Fourier expansions, i.e., 

$$
\texttt{FE}_{c(\neg x_1,\cdots,x_n)}(l_1,\cdots,l_n) = \texttt{FE}_{c(x_1,\cdots, x_n)}(-l_1,\cdots,l_n).
$$
\begin{itemize}
    \item 
    \textbf{CNF clause (OR)}: the Fourier expansion of an clause is
    $$\texttt{FE}_{OR(x_1,\cdots,x_k)}(X) = 1 + \prod_{i\in[k]} \frac{(1+x_i)}{2}$$ 
    
    Therefore $\texttt{FE}_{OR(x_1,\cdots,x_k)}(l)=0$ iff at least element in $l$ takes value $-1$. Rounding such a real point gives an assignment with a positive literal, which satisfies the clause.

\item     \textbf{XOR}: The Fourier expansion of an XOR constraint is $$\texttt{FE}_{XOR(x_1,\cdots,x_k)}=\frac{1+\prod_{i\in [k]}x_i}{2}.$$

If $\texttt{FE}_{XOR(x_1,\cdots,x_k)}(l)=0$, then it means that the product of all literals is negative, rounding which gives an assignment satisfying the XOR constraint.

\item    \textbf{NAE}: The Fourier expansion of an NAE constraint is

\begin{equation*}
\begin{split}
\texttt{FE}_{NAE(x_1,\cdots,x_k)}&=\frac{(1+x_1)\cdots(1+x_k)+(1-x_1)\cdots(1-x_k)}{2^k}\\
&=\frac{2+2x_1x_2+\cdots+2x_{k-1}x_k+2x_1x_2x_3x_4+\cdots}{2^k}.  
\end{split}    
\end{equation*}

If all elements in a real vector $a$ have the same sign, then $\texttt{FE}_{NAE}(l)>0$. Therefore, if $\texttt{FE}_{NAE}(l)=0$, there must exist at least two variables with different value sign. Rounding such a real vector gives a solution. 
   
\item     \textbf{Pseudo-Boolean constraints}: Consider the example  $c=x_1\wedge x_2$ (can be also viewed as a cardinality constraint). Its Fourier expansion is
$$
  \text{FE}_{\mathrm{AND}}(x_1,x_2)=
  \tfrac{3}{4}+
  \tfrac{1}{4}x_1+
  \tfrac{1}{4}x_2-
  \tfrac{1}{4}x_1 x_2.
$$
 $l=(3,3)$ yields
$
  \text{FE}_{\mathrm{AND}}(3,3)
  =\tfrac{3}{4} + \tfrac{3}{4} + \tfrac{3}{4} - \tfrac{9}{4}
  = 0.
$
However, rounding $(3,3)$ to $(+1,+1)$ corresponds to $(\texttt{False},\texttt{False})$, which does not satisfy $x_1\land x_2.$
\end{itemize}
\end{proof}

\subsection{Proof of Theorem~\ref{prop:left}}\label{app:proof_left}

\noindent\textbf{Theorem~\ref{prop:left} (restated).}
\textit{If a constraint is rounding-friendly, then it must have isolated violations.}

Before proving Theorem~\ref{prop:left}, we define the notation of partial assignment to Boolean functions.

\begin{definition}[partial assignment] For a subset of variables $Y\subseteq X$ and an assignment $b$ of $Y$, i.e., $b: Y\to \{True,False\}$, let $c$ be a constraint, then $c|_{Y\gets b}$ is the partially assigned Boolean constraint where all variables in $Y$ are fixed to the values in $b$.     
\end{definition}
We also need the following proposition for decomposing the Fourier expansion on a variable.
\begin{proposition}[Splitting on a variable] Let $F$ be a multilinear polynomial on a set of variables $X$. Then for each variable $x\in X$, there exists multilinear polynomials $G$ and $H$ on $X\setminus x_i$ such that
$
F(X) = G(X\setminus x_i) x_i + G(X\setminus x_i).
$ 
\label{prop:split}
\end{proposition}

Finally, we need the following auxiliary proposition, which indicates that the value of a non-constant multilinear polynomial can cover the whole real domain. 

\begin{proposition} If a multilinear polynomial $F$ is not constant, then for every $M\in\mathbb{R}$, there exists a real vector $l\in \mathbb{R}^n$, such that $F(l)=M$.
\label{prop:anyRealNumber}
\end{proposition}

\begin{proof}
    Since $F$ is not constant, then by Proposition \ref{prop:split} there exists $i\in[n]$ such that $F$ can be rewritten as 
    $$
    F(X) = x_i\cdot G(X\setminus x_i) + H(X\setminus x_i),
    $$

   where $f(X\setminus x_i)\not\equiv 0$ (otherwise $x$ can be removed from the variable set of $F$). Since  $f(X\setminus x_i)\not\equiv 0$, there exists a real vector $l'\in \mathbb{R}^{n-1}$ such that $f(X\setminus x)(l')\neq 0$. Construct a real vector $l\in \mathbb{R}^n$ which assigns $x_i$ with $\frac{M-G(l')}{H(l')}$ and the rest of variables assigned by $l'$, we have $F(l)=M$. 
\end{proof}

Now we are ready to prove Theorem \ref{prop:left}.
\begin{proof}
First, notice that if a constraint is rounding-friendly for some $\epsilon>0$, then it must be $0$-rounding-friendly. Equivalently, we show that if a constraint has violations with hamming distance $1$, then it is not $0$-rounding-friendly. 

We assume $c$ has at least two variables otherwise the statement is trivial. If $c$ has two violations with hamming distance $1$, then those two violations only differ in the value of one variable, say $x_n$. Let the value of those two violations on $x_1,\cdots,x_{n-1}$ be $b$ such that $c|_{x_1,\cdots,x_{n-1}\gets b}\equiv \texttt{False}$. Consider the sequence of partially assigned constraint, $c_0=c, c_1=c|_{x_1\gets b}, c_2=c|_{x_1,x_{2}\gets b}, \cdots, c_{n-1}=c|_{x_1,x_{2},\cdots, x_{n-1}\gets b}\equiv \texttt{False}$. Since $c_0=c$ is not constant and $c_{n-1}\equiv \texttt{False}$, there exists a minimum $k$ such that $c_{k-1}$ is not constant but $c_k\equiv \texttt{False}$. 

Without loss of generality, we assume $b_k$ = \texttt{True}. By the Shannon expansion of a Boolean function, we have

\begin{equation*}
\begin{split}
c_{k-1} &= (x_k\wedge c_{k-1}|_{x_k\gets \texttt{True}}) \vee (\neg x_k \wedge c_{k-1}|_{x_k\gets \texttt{False}})\\
&=(x_k\wedge c_k)\vee (\neg x_k\wedge c_{k-1}|_{ x_k\gets \texttt{False}})\\
&=\neg x_k\wedge c_{k-1}|_{x_k\gets \texttt{False}}.
\end{split}
\end{equation*}

With the expression above, $\texttt{FE}_{c_{k-1}}$ can be written as $$\texttt{FE}_{c_{k-1}}=(x_k+1)\cdot g(x_{k+1},\cdots,x_{n})+1,$$ 
 where the multilinear polynomial $g=\texttt{FE}_{c_{k-1}|_{x_k\gets 1}}$. Note that $\texttt{FE}_{c_{k-1}}(x_k\gets -1)=1$ as expected.
 Note that $g(x_{k+1},\cdots,x_{n})\not\equiv 0$, otherwise  $c_{k-1}\equiv \texttt{False}$, which contradicts with the assumption that  $k$ is minimum. Then by Proposition \ref{prop:anyRealNumber}, there exists a real vector assignment to $x_{k+1},\cdots, x_n$, denoted by $a'$, such that $g(a')=1$. Then we assign $x_k$ to $-2$ and have 
 
 $$\texttt{FE}_{c_{k-1}}(x_k\gets -2, x_{k+1},\cdots,x_n\gets a')=(-2+1)\cdot 1+1=0.$$ 
 
 Therefore $a'' = (b_1,\cdots,b_{k-1}, -2, a')$ is a solution (zero point) of $\texttt{FE}_{c}$. Rounding this point gives $b'=\texttt{sgn}(a'')=(b_1,\cdots,b_{k-1}, -1, \texttt{sgn}(a'))$.  Since $b_k=\texttt{True}$, i.e., $-1$, then we have $c(b')=c_{k}(\texttt{sgn}(a'))=\texttt{False}$, which means rounding a solution $a''$ of $\texttt{FE}_c$ in the real domain gives an assignment $b'$ which violates $c$. Therefore, $c$ is not rounding-friendly.
\end{proof}

\subsection{Proof of Proposition~\ref{prop:CNFW}}\label{app:proof_CNFW}

\noindent\textbf{Proposition~\ref{prop:CNFW} (restated).}
An \texttt{OR} clause of length $k$ is  $\frac{1}{2^k}$-rounding-friendly. An XOR constraint is $\frac{1}{2}$-rounding-friendly. An NAE of length $k$ is $\frac{1}{2^{k-1}}$-rounding-friendly.

\begin{proof}
We discuss separately for each type of constraints. Without loss of generality, we again assume all literals are positive.

    If $c$ is an OR clause, then its Fourier expansion is $\prod_i (\frac{1+x_i}{2})$. If $|\prod_i (\frac{1+l_i}{2})|< \frac{1}{2^k}$, then there exists $x_i$ such that $\frac{l_i+1}{2}<\frac{1}{2}$, i.e., $l_i<0$. Therefore, rounding the value of $x_i$ gives an assignment that satisfies the constraint.

    If $c$ is an XOR constraint, then its Fourier Expansion is $\frac{1-\prod_i x_i}{2}$. If $\frac{1-\prod_i l_i}{2}<\frac{1}{2}$, then $\prod_i l_i<0$, which indicates that rounding the real assignment gives a solution.

If $c$ is an NAE constraint, then its Fourier expansion is 

\begin{equation*}
\begin{split}
\texttt{FE}_{NAE(x_1,\cdots,x_k)}&=\frac{(1+x_1)\cdots(1+x_k)+(1-x_1),\cdots(1-x_k))}{2^k}\\
&=\frac{2+2x_1x_2+\cdots+x_{k-1}x_k+2x_1x_2x_3x_4+\cdots}{2^k}.
\end{split}
\end{equation*}

If $FE_{NAE}(l)<\frac{1}{2^{k-1}}$, then there must be at least two variables whose value are with different \texttt{sgn}, which indicate that rounding $l$ gives a solution.
\end{proof}

\noindent\textbf{Corollary~\ref{coro:bound} (restated).}
 Suppose $F_{C,\alpha=0}^{sq}(l)=W$ for $l\in \mathbb{R}^n$. If the formula  $f$ is k-CNF , then \texttt{sgn}$(l)$ violates at most $\ceil{4^k\cdot W}$ constraints. If $f$ is pure-\texttt{XOR}, then the rounded assignment violates at most $\ceil{4W}$ constraints. If $f$ is k-\texttt{NAE}, then \texttt{sgn}$(l)$ violates $\ceil{4^{k-1}W}$ constraints.

\begin{proof}
    We prove that in general, if all constraints in the formula are $\epsilon$-rounding-friendly, then $F_{C,\alpha=0}^{sq}(l)=M$ indicates that $\texttt{sgn}(l)$ violates at most $\ceil{\frac{M}{\epsilon^2}}$ constraints.

   Since $F_{C,\alpha=0}^{sq}(l)=M$, then there are at most $\ceil{\frac{M}{\epsilon^2}}$ constraints $c$ s.t. $\texttt{FE}_c(l)\ge\epsilon$. For all constraints $c$ s.t. $\texttt{FE}_c(l)<\epsilon$, they are satisfied due to $\epsilon$-rounding-friendly. Therefore the number of violated constraints is upper-bounded by $\ceil{\frac{M}{\epsilon^2}}$.
\end{proof}

\iffalse
\section{Additional Experimental Results}\label{sec:itemStyles}

\begin{figure}[!ht]
    \centering
    \begin{tabular}{cc}
        SLSQP-SQUARE & SLSQP-ABS
        \\
        \includegraphics[width=.4\linewidth]{XOR_1000_penaltyTermSQUARE.png}
        & 
        \includegraphics[width=.4\linewidth]{XOR_1000_penaltyTermABS.png}
        \\
        SLSQP-SQUARE-C & SLSQP-ABS-C
        \\
        \includegraphics[width=.4\linewidth]{XOR_1000_penaltyTermSQUARE-C.png}
        &
        \includegraphics[width=.4\linewidth]{XOR_1000_penaltyTermABS-C.png}
    \end{tabular}
    \caption{Effect of the penalty term from \{0,0.2,0.4,0.6,0.8\} on random 2-XOR instances}
    \label{fig:valueandgradient}
\end{figure}
\fi
\end{document}